\documentclass[12pt]{article}
\usepackage{amsfonts}
\usepackage{amsmath}
\usepackage{amssymb}
\usepackage{graphicx}
\usepackage{color}
\usepackage[all, knot]{xy}
\usepackage{tikz}

\usepackage[utf8]{inputenc}
\usepackage{epstopdf}
\usepackage[footnotesize]{caption}
\usepackage{amsthm}
\usepackage{enumitem}
\usepackage{mathrsfs}

\usepackage[margin=3cm]{geometry}

\def \be {\begin{equation}}
\def \ee {\end{equation}}
\def \bea {\begin{eqnarray}}
\def \eea {\end{eqnarray}}
\def \nn {\nonumber}

\def \rr {\raise.35ex\hbox{\small $\prime$}\kern-.17em{\mbox{\large $\imath$}}}

\def \dels {\partial\kern-.6em /\kern.1em}
\def \As {{A\kern-.5em / \kern.5em}}
\def \Ds {D\kern-.7em / \kern.5em}

\def \ks {k\kern-.5em /}
\def \ls {l\kern-.5em /}

\newcommand{\ci}[1]{}

\newcommand{\ba}{\begin{eqnarray}}
\newcommand{\ea}{\end{eqnarray}}
\newcommand{\bal}{\begin{align}}
\newcommand{\eal}{\end{align}}
\newcommand{\bay}[1]{\left(\begin{array}{#1}}
\newcommand{\eay}{\end{array}\right)}

\newcommand{\hide}[1]{}

\newtheorem{theorem}{Theorem}

\newtheorem{lemma}{Lemma}

\newlist{axioms}{enumerate}{2}
\setlist[axioms,1]{label=\textbf{A\arabic{axiomsi}.}, ref=A\arabic{axiomsi}}
\setlist[axioms,2]{label=\textbf{A\arabic{axiomsi}\rlap{\myEnumCounter{axiomsii}}.},%
                   ref=A\arabic{axiomsi}\myEnumCounter{axiomsii},%
                   align=parleft,%
                   leftmargin=0em,%
                   itemsep=1.4ex,%
                   before={\stepcounter{axiomsi}}}
                   
  \usetikzlibrary{decorations.markings}

\begin{document}

\begin{titlepage}

\begin{center}

\hfill
\vskip .2in

\textbf{\LARGE
Bell's Inequality and Entanglement in Qubits
\vskip.3cm
}

\vskip .5in
{\large
Po-Yao Chang$^a$ \footnote{e-mail address: pychang@physics.rutgers.edu}, Su-Kuan Chu$^{b,c}$ \footnote{e-mail address: skchu@terpmail.umd.edu} and Chen-Te Ma$^d$ \footnote{e-mail address: yefgst@gmail.com} 
\\
\vskip 1mm
}
{\sl
$^a$
Center for Materials Theory, Rutgers University, Piscataway, New Jersey, 08854, 
\\
$^b$
Joint Quantum Institute, NIST/University of Maryland, College Park,\\
 Maryland 20742,
\\
$^c$
Joint Center for Quantum Information and Computer Science,\\
NIST/University of Maryland, College Park, Maryland 20742, USA.
\\
$^d$
Department of Physics and Center for Theoretical Sciences, \\
National Taiwan University,\\ 
Taipei 10617, Taiwan, R.O.C.
}\\
\vskip 1mm
\vspace{40pt}
\end{center}
\begin{abstract}
We propose an alternative evaluation of quantum entanglement by measuring the maximum violation of the Bell's inequality without performing a partial trace operation. This proposal is demonstrated by bridging the maximum violation of the Bell's inequality and the concurrence of a pure state in an $n$-qubit system, in which one subsystem only contains one qubit and the state is a linear combination of two product states. We apply this relation to the ground states of four qubits in the Wen-Plaquette model and show that they are maximally entangled. A topological entanglement entropy of the Wen-Plaquette model could be obtained by relating the upper bound of the maximum violation of the Bell's inequality to the concurrences of a pure state with respect to different bipartitions.
\end{abstract}

\end{titlepage}

\section{Introduction}
\label{1}
Entanglement measurements provide a way to extract quantum information from many-body wave-functions \cite{Steane:1997kb}. The most significant measure of entanglement is given by the entanglement entropy, 
$S_A=-{\rm Tr } \rho_A \ln \rho_A$ with $\rho_A={\rm Tr_B \rho}$ being the reduced density matrix of the subsystem $A$, and $\rho$ being the density matrix of a Hilbert space $\mathbb{H}=\mathbb{H}_A\otimes\mathbb{H}_B$. In other words, the entanglement entropy characterizes the entanglement between two complementary subsystems $A$ and $B$. The entanglement entropy has been observed experimentally in a two-qubit system, but measuring the entanglement entropy for a higher-qubit system is still under development.

On the other hand, a qualitative detection of quantum entanglement 
could be performed experimentally by the observation of the violation of the Bell's inequality \cite{Clauser:1969ny}.
The original theorem, proposed by the John S. Bell \cite{Bell:1964kc}, states that correlations between the outcomes of different measurements of two separated particles must satisfy the inequality under the local realism.
The violation of the constraints (the Bell's inequality) indicates the quantum effect of correlations or "entangledness" in quantum systems, 
which could be presented in two-qubit systems theoretically \cite{Cirelson:1980ry}. 
Although the violation of the Bell's inequality may not reveal the general structure of entanglement of a quantum state, the relation between the entanglement, measured in terms of the concurrence \cite{Bennett:1996gf}, 
and the violation of the Bell's inequality was shown in two-qubit systems \cite{Verstraete:2001, Horodecki:1995}.
The generalization for higher-qubit systems is still unclear.   

In this letter, we discuss relations between the maximum violation of the Bell's inequality of an $n$-qubit Bell's operator \cite{Gisin:1998ze} and the concurrence of a pure state when the $i$-th qubit operators in the Bell's operator are $\bf{n}\cdot\boldsymbol{\sigma}$, where $\bf{n}$ is a unit vector and $\boldsymbol{\sigma}$ are Pauli matrices.
One crucial point is that the quantum entanglement depends on a partial trace operation in a system, but the Bell's inequality does not. At first glance, this suggests that a quantitative entanglement measure by the Bell's inequality is difficult.
Thus, bridging the maximum violation of the Bell's inequality and measures of quantum entanglement provides a huge application of an entanglement measure without performing a partial trace operation to detect the entanglement quantities.

There are various $n$-qubit systems exhibiting topological properties such as the toric code model \cite{Kitaev:1997wr} and the Wen-Plaquette model \cite{Wen:2003yv}.
One of the topological signature is that the total quantum dimension of quasi-paritcles could be detected
from the universal term in the entanglement entropy \cite{Hamma:2005zz}, i.e., topological entanglement entropy \cite{Kitaev:2005dm,Grover:2013ifa}. This motivates us to apply our theorem to the Wen-Plaquette model. 
We find that the upper bound of the maximum violation of the Bell's inequality in the Wen-Plaquette model indicates that the ground state is maximally entangled. The use of
the maximally entangled property for a six-qubit state in the Wen-Plaquette model could relate to the topological entanglement entropy via the maximum violation of the Bell's inequality.

\section{Entanglement and Maximum Violation}
\label{2}
A Bell's operator of $n$ qubits is defined iteratively as ${\cal B}_n$ \cite{Gisin:1998ze}:
${\cal B}_n={\cal B}_{n-1}\otimes\frac{1}{2}\bigg(A_n+A_n^{\prime}\bigg)+{\cal B}^{\prime}_{n-1}\otimes\frac{1}{2}\bigg(A_n-A_n^{\prime}\bigg),$
where $A_n={\bf a}_n \cdot \boldsymbol{\sigma}$ and $A^{\prime}_n={\bf a}_n^{\prime} \cdot \boldsymbol{\sigma}$ are the operators in the $n$-th qubit
with ${\bf a}_n$ and ${\bf a}_n^{\prime}$ being unit vectors and $\boldsymbol{\sigma}=(\sigma_x,\sigma_y,\sigma_z)$ being the Pauli matrices. The operators $\frac{1}{2} {\cal B}_{n-1}$ and $\frac{1}{2} {\cal B}_{n-1}^{\prime}$ act on the rest of the qubits.
Notice that we choose $\frac{1}{2}{\cal B}_1={\bf b} \cdot \boldsymbol{\sigma}$ and $\frac{1}{2}{\cal B}_1^{\prime}={\bf b}^{\prime} \cdot \boldsymbol{\sigma}$
with ${\bf b}$ and ${\bf b}^{\prime}$ being unit vectors.
It is known that for a $n$-qubit system, the upper bound of the expectation value of 
the Bell operator $\mbox{Tr}(\rho{\cal B}_n)\le 2^{\frac{n+1}{2}}$ \cite{Gisin:1998ze} leads to the violation of the Bell-CHSH inequality \cite{Clauser:1969ny}.

For a given density matrix $\rho$, 
the maximum expectation value of a Bell's operator is referred to as the maximum violation of the Bell's inequality. 
Here we demonstrate a relation between the maximum violation of the Bell's inequality and the concurrence (an entanglement quantity) in an $n$-qubit system when the all $i$-th operators in the Bell's operator are $A_i$ and $A_i^{\prime}$ for $2\le i< n$:
\bea
\tilde{{\cal B}}_n&=&{\cal B}_1\otimes A_2\otimes A_3\cdots\otimes A_{n-2}\otimes A_{n-1}\otimes\frac{1}{2}\bigg(A_n+A_n^{\prime}\bigg)
\nn\\
&&+{\cal B}^{\prime}_1\otimes A_2^{\prime}\otimes A_3^{\prime}\cdots\otimes A_{n-2}^{\prime}\otimes A_{n-1}^{\prime}\otimes\frac{1}{2}\bigg(A_n-A_n^{\prime}\bigg).
\label{Eq:mBell}
\eea

To proceed our derivation, we introduce an $R$-matrix: $R_{i_1i_2\cdots i_n}\equiv\mbox{Tr}(\rho\sigma_{i_1}\otimes\sigma_{i_2}\otimes\cdots\otimes\sigma_{i_n})\equiv R_{Ii_n}$,
where $\rho$ is a density matrix, $\sigma_{i_\alpha}$ is the Pauli matrix with $i_\alpha=x,y,z$ and $\alpha=1, 2, \cdots, n$ are the site indices. We express the $R$-matrix as a $3^{n-1} \times 3$ matrix $R_{I i_n}$ with the first index being $I=i_1i_2\cdots i_{n-1}$ and the second index being $i_n$. In a two-qubit system, the maximum violation of the Bell's inequality
is computed from a $3 \times 3$ matrix $R_{ij}$ defined above \cite{Horodecki:1995}.
Now we generalize the maximum violation of the Bell's inequality ($\tilde{{\cal B}}_n$) in a $n$-qubit system
by using the $R$-matrix.
\begin{lemma}
\label{me}
The maximum violation of the Bell's inequalities $\gamma\equiv\max_{\tilde{{\cal B}}_n} {\rm Tr}(\rho{\tilde{\cal B}}_n)\le2\sqrt{u_1^2+u_2^2}$, where $u^2_1$ and $u^2_2$ are the first two largest eigenvalues of $R^\dagger R$ when $n>2$ and $\gamma=2\sqrt{u_1^2+u_2^2}$ when $n=2$.
\end{lemma}
\begin{proof}[Proof]
We first introduce two three-dimensional orthonormal vectors $\hat{c}$ and $\hat{c'}$ such that $\hat{a}+\hat{a'}=2\hat{c} \cos \theta$
and $\hat{a}-\hat{a'}=2 \hat{c'} \sin \theta$, where $\theta\in\lbrack 0,\frac{1}{2}\pi\rbrack$, through three-dimensional unit vectors $\hat{a}$ and $\hat{a'}$. The maximum violation of the Bell's inequality is defined as $\gamma\equiv\max_{\tilde{{\cal B}}_n} \mbox{Tr}(\rho\tilde{{\cal B}}_n)$
with the Bell's operator of the $n$-qubit $\tilde{{\cal B}}_n$ defined in \eqref{Eq:mBell}.
By using the $R$-matrix and the unit vectors $\hat{B}\equiv\hat{B}_I=\hat{B}_{i_1i_2\cdots i_{n-1}}\equiv\hat{a}_{1, i_1}\hat{a}_{2, i_2}\cdots\hat{a}_{n-1, i_{n-1}}$, $\hat{B'}\equiv\hat{B'}_I=\hat{B'}_{i_1i_2\cdots i_{n-1}}\equiv\hat{a'}_{1, i_1}\hat{a'}_{2, i_2}\cdots\hat{a'}_{n-1, i_{n-1}}$, $\hat{a}\equiv\hat{a}_{n, i_n}$, and $\hat{a'}\equiv\hat{a^{\prime}}_{n, i_n}$, in which ${\hat{B}}$ and $\hat{B^{\prime}}$ are unit vectors in $3^{n-1}$ dimensions, we have
$\gamma
=\max_{\hat{B},\hat{B'},\hat{a},\hat{a'}}  \bigg(\langle\hat{B}, R (\hat{a}+\hat{a'})\rangle+\langle\hat{B'}, R (\hat{a}-\hat{a'})\rangle\bigg)   \le\max_{\hat{c}, \hat{c'},\theta} \bigg( 2||R \hat{c}|| \cos \theta +2||R \hat{c'}|| \sin \theta\bigg)  
= 2\sqrt{u_1^2+u_2^2}$, 
in which $u^2_1$ and $u^2_2$ are the first two largest eigenvalues of $R^\dagger R$.
The inner product and the norm are defined as $\langle P, Q\rangle\equiv P^{\dagger}Q$ and $||U||\equiv \sqrt{U^{\dagger}U}$. Because $R(\hat{a}+\hat{a^{\prime}})$ and $R(\hat{a}-\hat{a^{\prime}})$ are defined in the $3^{n-1}$ dimensions and each unit vector $\hat{B}$ and $\hat{B^{\prime}}$ only contains $2(n-1)$ parameters, it could not guarantee that $\hat{B}$ parallels $R(\hat{a}+\hat{a^{\prime}})$ and $\hat{B^{\prime}}$ parallels $R(\hat{a}-\hat{a^{\prime}})$, except for $n=2$.
\end{proof}
An earlier approach to relate the maximum violation of the Bell's inequality
and the concurrence of a pure state $C(\psi)\equiv\sqrt{2(1-{\rm Tr} \rho_A^2)}$ \cite{Bennett:1996gf} in a two-qubit system is discussed in \cite{Verstraete:2001}.

We generalize the relation of the maximum violation of the Bell's inequality
and the concurrence of a pure state in an $n$-qubit system when a state is a linear combination of two product states.
The concurrence is computed with respect to the bipartition with ($n-1$) qubits in subsystem $B$ and one qubit in subsystem $A$.
\begin{theorem}
\label{mec}
For an $n$-qubit state $|\psi \rangle=|u\rangle_B\otimes\big(\lambda_+|v\rangle_B \otimes |1\rangle_A +\lambda_-|\tilde{v}\rangle_B \otimes |0\rangle_A\big)$
with $\lambda_+|v\rangle_B \otimes |1\rangle_A +\lambda_-|\tilde{v}\rangle_B \otimes |0\rangle_A$ being a non-biseparable, $\alpha$-qubit state,
$|u\rangle_B$, $|v\rangle_B$, $|\tilde{v}\rangle_B$ being product states, and
the maximum violation of the Bell's inequality in an $n$-qubit system is $\gamma=2f_{\alpha}(\psi)$,
in which the function $f_{\alpha}(\psi)$ is defined as:\\
$(1)$ $\alpha$  is an even number:
\begin{align}
&f_{\alpha}(\psi)\equiv\sqrt{1+2^{\alpha-2}C^2(\psi)},& \qquad    &2^{2-\alpha}\ge C^2(\psi), \nn\\
&f_{\alpha}(\psi)\equiv2^{\frac{\alpha-1}{2}}C(\psi),& \qquad       &2^{2-\alpha}\le C^2(\psi).
\end{align}
$(2)$ $\alpha$  is an odd number:
\begin{align}
&f_{\alpha}(\psi)\equiv\sqrt{1+\big(2^{\alpha-2}-1\big)C^2(\psi)},& \qquad  &\frac{1}{1+2^{\alpha-2}}\ge C^2(\psi),\nn\\
&f_{\alpha}(\psi)\equiv2^{\frac{\alpha-1}{2}}C(\psi),& \qquad &\frac{1}{1+2^{\alpha-2}}\le C^2(\psi).
\end{align}
Here, $C(\psi)$ is the concurrence of a pure state computed with respect to the bipartition that subsystem $B$ contains $(n-1)$ qubits
and subsystem $A$ contains one qubit. 
\end{theorem}
\begin{proof}[Proof]
The Hilbert space for an $n$-qubit system is bipartitioned as $\mathbb{H}=\mathbb{H}_B \otimes \mathbb{H}_A$,
in which dimensions of the sub-Hilbert spaces are $\dim(\mathbb{H}_A)=2$ and $\dim(\mathbb{H}_B)=2^{n-1}$.
We consider a pure state with respect to this bipartition 
$|\psi \rangle=|u\rangle_B\otimes(\lambda_+|v\rangle_B \otimes |1\rangle_A +\lambda_-|\tilde{v}\rangle_B \otimes |0\rangle_A)$,
where $|u\rangle_B\otimes|v\rangle_B$ and $|u\rangle_B\otimes|\tilde{v}\rangle_B$ are the product states in $\mathbb{H}_B$
and $|1\rangle_A$ and $|0\rangle_A$ are the states in $\mathbb{H}_A$. By using the property ${\rm Tr}\rho_A=\lambda_+^2 +\lambda_-^2=1$
and $C(\psi)=\sqrt{2(1-\lambda_+^4-\lambda_-^4)}$, the coefficients $\lambda_{\pm}$ can be expressed in terms of the concurrence, $\lambda_{\pm}^2=\big(1\pm\sqrt{1-C^2(\psi)}\big)/2$.
The matrix elements of the $R$-matrix are\\
\bea
R_{Ix}&=&\lambda_+\lambda_- {\rm Tr}  \left [ \bigotimes_{I_1} \sigma_{I_ 1}| u \rangle \langle u |     
\bigotimes_{I_2} \sigma_{I_ 2} (|v \rangle \langle \tilde{v}|+|\tilde{v} \rangle \langle v|)  \right],  \notag\\
R_{Iy}&=&-i \lambda_+\lambda_- {\rm Tr}  \left [ \bigotimes_{I_1} \sigma_{I_ 1}| u \rangle \langle u |     
\bigotimes_{I_2} \sigma_{I_ 2} (|v \rangle \langle \tilde{v}|-|\tilde{v} \rangle \langle v|)  \right],  \notag\\
R_{Iz}&=& -\lambda_+^2{\rm Tr}  \left [ \bigotimes_{I_1} \sigma_{I_ 1}| u \rangle \langle u |     
\bigotimes_{I_2} \sigma_{I_ 2} |v \rangle \langle {v}|  \right] 
 +\lambda_-^2{\rm Tr}  \left [ \bigotimes_{I_1} \sigma_{I_ 1}| u \rangle \langle u |     
\bigotimes_{I_2} \sigma_{I_ 2} |\tilde{v} \rangle \langle \tilde{v}|  \right],
\nn\\
\eea
where $I\equiv I_1I_2$. Here we choose the basis that $|0\rangle\equiv(1,0)^{\rm T}$
and $|1\rangle\equiv(0,1)^{\rm T}$.

The conditions for non-vanishing matrix elements $R_{Ix}$ are
$(n-\alpha)$  number of $\sigma_z$ matrices in $I_1$, $(\alpha-1-i)$ number of $\sigma_x$ matrices and $i$ number of $\sigma_y$
matrices in $I_2$ with $i$ being an even integer.
The conditions for non-vanishing matrix elements $R_{Iy}$ are
$(n-\alpha)$  number of $\sigma_z$ matrices in $I_1$, $(\alpha-1-j)$ number of $\sigma_x$ matrices and $j$ number of $\sigma_y$
matrices in $I_2$ with $j$ being an odd integer.
The conditions for non-vanishing matrix elements $R_{Iz}$ are
$(n-\alpha)$  number of $\sigma_z$ matrices in $I_1$, $(\alpha-1)$ number of $\sigma_z$ 
matrices in $I_2$. 

The above conditions lead to the diagonal form of the matrix $R^\dagger R$.
In the case that $\alpha$ is an even integer, the set of eigenvalues of $R^\dagger R$ is $\{2^{\alpha-2}C^2(\psi),  2^{\alpha-2}C^2(\psi),  1 \}$.
In the case that $\alpha$ is an odd integer, the set of eigenvalues of $R^\dagger R$  is $\{2^{\alpha-2}C^2(\psi),  2^{\alpha-2}C^2(\psi),  1-C^2(\psi) \}$.

Now we want to show $\gamma=\max_{\hat{B},\hat{B'},\hat{a},\hat{a'}}  \langle\hat{B}, R (\hat{a}+\hat{a'})\rangle+\langle\hat{B'}, R (\hat{a}-\hat{a'})\rangle   
= 2\sqrt{u_1^2+u_2^2}$, in which $u^2_1$ and $u^2_2$ are the first two largest eigenvalues of $R^\dagger R$, $\hat{B}\equiv\hat{a}_{1, i_1}\hat{a}_{2, i_2}\cdots\hat{a}_{n-1, i_{n-1}}$, $\hat{B'}\equiv\hat{a'}_{1, i_1}\hat{a'}_{2, i_2}\cdots\hat{a'}_{n-1, i_{n-1}}$, $\hat{a}\equiv\hat{a}_{n, i_n}$ and $\hat{a^{\prime}}\equiv\hat{a^{\prime}}_{n, i_n}$, where $\hat{a}_{n, i_n}+\hat{a^{\prime}}_{n ,i_n}\equiv2\hat{c}_{n, i_n}\cos\theta$, $\hat{a}_{n, i_n}-\hat{a^{\prime}}_{n ,i_n}\equiv2\hat{c^{\prime}}_{n, i_n}\sin\theta$, and $\theta\in\lbrack 0, \pi/2\rbrack$. 
This equality holds when $\hat{B}$ parallels $R(\hat{a}+\hat{a^{\prime}})$ and $\hat{B'}$ parallels $R(\hat{a}-\hat{a^{\prime}})$.
One natural choice of $a_{\alpha,  i_{\alpha}}$ and $a_{\alpha,^\prime  i_{\alpha}}$ could be obtained by equating two ratios,
$\big|R_{Ix}(\hat{a}_x + \hat{a^{\prime}}_x)/R_{I'y}(\hat{a}_y+\hat{a^{\prime}}_y)\big|=\big|B^{}_I/B^{}_{I'}\big|$
and $\big|R_{Ix}(\hat{a}_x-\hat{a^{\prime}}_x)/R_{I'y}(\hat{a}_y-\hat{a^{\prime}}_y)\big|=\big|B^{\prime}_I/B^{\prime}_{I'}\big|$, where $I$ and $I'$ is chosen in a way that one site of the $I_2$ in $I$ is labeled by $x$ and in $I'$ is labeled by $y$, and other sites of the $I_2$ in $I$ and $I^{\prime}$ are labeled by the same symbols.
This leads to $\big|\hat{a}_{I, x}/\hat{a}_{I^{\prime}, y}\big|=\big|\hat{c}_{n, x}/\hat{c}_{n, y}\big|$ 
and $\big|\hat{a^{\prime}}_{I, x}/\hat{a^{\prime}}_{I^{\prime}, y}\big|=\big|\hat{c^{\prime}}_{n, x}/\hat{c^{\prime}}_{n, y}\big|.$
When $u^2_1=(R^{\dagger}R)_{xx}$ and $u^2_2=(R^{\dagger}R)_{yy}$, we could choose 
$(\hat{c}_{n, x}, \hat{c}_{n, y}, \hat{c}_{n, z})^{\rm T}=\big(1/\sqrt{2}\big)(1, 1, 0)^{\rm T}$, $(\hat{c^{\prime}}_{n, x}, \hat{c^{\prime}}_{n, y}, \hat{c^{\prime}}_{n, z})^{\rm T}=\big(1/\sqrt{2}\big)(1, -1, 0)^{\rm T}$, 
and $\cos(\theta)=\sin(\theta)=\sqrt{2}/2$ to show $\gamma=2\sqrt{u_1^2+u_2^2}$. For the other case $u^2_1=(R^{\dagger}R)_{zz}$ and $u^2_2=(R^{\dagger}R)_{xx}$, we could choose 
$(\hat{c}_{n, x}, \hat{c}_{n, y}, \hat{c}_{n, z})^{\rm T}=(0, 0, 1)^{\rm T}$, $(\hat{c^{\prime}}_{n, x}, \hat{c^{\prime}}_{n, y}, \hat{c^{\prime}}_{n, z})^{\rm T}=\big(1/\sqrt{2}\big)(1, 1, 0)^{\rm T}$,
$\cos(\theta)=u_1/\sqrt{u_1^2+u_2^2}$ and $\sin(\theta)=u_2/\sqrt{u_1^2+u_2^2}$ to prove $\gamma=2\sqrt{u_1^2+u_2^2}$.

\end{proof}
We demonstrate that the maximum violation of the Bell's inequality ($\tilde{\cal{B}}_n$) is directly related to the concurrence of the pure state
when the subsystem $A$ only contains one qubit and the state is a linear combination of two product states. 
For the maximally entangled state with the concurrence $C(\psi)=1$, the maximum violation of the Bell's inequality $\gamma=2^{\frac{\alpha+1}{2}}\leq 2^{\frac{n+1}{2}}$
satisfies the upper bound of the Bell's operator of an $n$-qubit system. Although we do not use the most generic form of the Bell's operator, the information of the state could be complete when the $n$-th qubit operators are measured. The extrapolation of the maximum violation of the Bell's inequality ($\tilde{\cal{B}}_n$) from the $R$-matrix could be equivalent to direct computing of the maximum violation of the Bell's inequality without losing generality.


\section{Applications to the Wen-Plaquette Model}
The Wen-Plaquette model \cite{Wen:2003yv} is defined by the Hamiltonian on a two-dimensional periodic lattice (torus) as
$H=\sum_{i} \sigma_x^i\sigma_y^{i+\hat{x}}\sigma_x^{i+\hat{x}+\hat{y}}\sigma_y^{i+\hat{y}}$,
in which qubits live on the vertices with the four-spin interaction on each plaquette.
A ground state of the Hamiltonian is an $n$-qubit state with $n$ being the number of vertices.
We first apply our Theorem to a four-qubit state, with the geometry of the system containing four vertices, 
eight edges, and four faces \cite{EN}.
There are four degenerate ground states $|G\rangle_{\rm 4-qubit} = \frac{1}{\sqrt{2}} (|0000\rangle + |1111 \rangle)$,
$ \frac{1}{\sqrt{2}} (|1010\rangle + |0101 \rangle)$, $ \frac{1}{\sqrt{2}} (|0011\rangle - |1100 \rangle)$, $ \frac{1}{\sqrt{2}} (|1001\rangle - |0110 \rangle)$.
The order of each site in these four-qubit states are defined in the Fig. \ref{F1} (a). 
Since the maximum violation of the Bell's inequalities for these ground states are $\gamma= 4\sqrt{2}$, the ground states have concurrence $C(\psi)=1$ according to the Theorem and are maximally entangled.

\begin{figure}
  \includegraphics[width=\columnwidth] {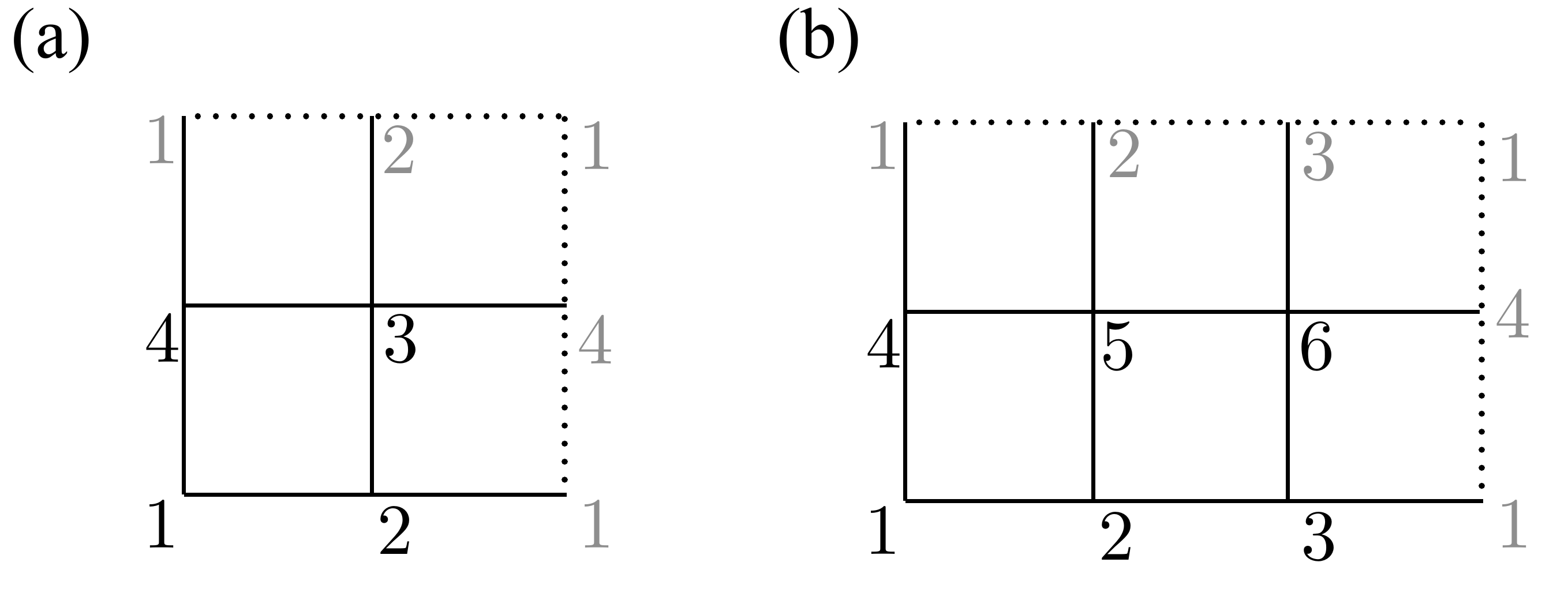}
  \caption{(a) A four-qubit state and (b) A six-qubit state for the Wen-Plaquette model on a torus.
  The right dashed line is identified as the left solid line and the top dashed line is identified as the bottom solid line in (a) and (b).
  The numbers are the site indices. The gray colored number is identified with the corresponding black colored number.}
  \label{F1}
\end{figure}

Before computing the upper bound of the maximum violation of the Bell's inequality of a ground state of the six-qubit in the Wen-Plaquette model, we consider a
six-qubit state,
$|G\rangle_{\rm 6-qubit}=\frac{\lambda_+}{\sqrt{2}}\big(-|111000\rangle +|001110\rangle \big)+\frac{\lambda_-}{\sqrt{2}}\big(|100011\rangle+|010101\rangle\big)$,  
with the site labels shown in Fig. \ref{F1} (b). We relate the upper bound of the maximum violation of the Bell's inequality to the concurrence of the state.
Two different bipartitions are considered: (1) subsystem $A$ contains site number six, and (2) subsystem $A$ contains sites number five and number six.
Here we use $\delta=1$ or $2$ as an indicator for the case one and the case two.
According to the Lemma,  we find that the upper bound of the maximum violation of the Bell's inequality could be expressed as a function of the concurrence of the pure state when we exchange the final site with the first site in the Bell's operator ($\tilde{\cal B}_n$) \bigg( We define
$C(\delta)\equiv\sqrt{2\bigg(1-2^{\delta-1}\mbox{Tr}\rho_{A(\delta)}^2\bigg)}$
as the concurrence of the pure state for the six-qubit state with respect to two different bipartitions.\bigg) 
(see Supplementary Material \cite{SM}). 

In the case that $\lambda_+=\lambda_-=1/\sqrt{2}$, the six-qubit state is a ground state of 
the Wen-Plaquette model and has the concurrence $C(\delta)=1$, which indicates the maximally entangled state. The entanglement entropy with respect to the two bipartitions are $S_{A(\delta=1)}=\ln 2$ and $S_{A(\delta=2)}=2\ln 2$, which could be obtained from the $R$-matrix through the inverse mapping, $\gamma\le 6=2\sqrt{13-2^{\delta+2}e^{-S_{A(\delta)}}}$.

In general, the entanglement entropy has a form $S_{A(L)}= \alpha L  -S_{\rm TEE}$,
in which the first term indicates the area law with $L$ being the length of the entangling boundary, $\alpha$ being a constant,
and $S_{\rm TEE}$ is called the topological entanglement entropy.
In the Wen-Plaquette model, the length of an entangling boundary $L$ is the number of bonds that connect subsystems $A$ and $B$.
We consider $L(\delta=1)=4$ and $L(\delta=2)=6$ to extract the area law of the entanglement entropy 
and obtain the topological entanglement entropy, $S_{\rm TEE}=\ln 2=\ln\sqrt{D}$, where $D=4$ is the number of distinct quasiparticles  \cite{Kitaev:2005dm, Grover:2013ifa}. 
Here, we demonstrate an indirect measure of the topological entanglement entropy by measuring the $R$-matrix.

\section{Outlook}
\label{5}
Recently, the ground states in the toric code model with three sites have realized in \cite{Li2016} by using a $^{13}$C-labeled trichloroethylene molecule. 
The ground states in the Wen-Plaquette model with four sites were also measured in the Iodotrifluroethylene (C$_2$F$_3$I) \cite{Luo:20161, Luo:20162} by using geometric algebra procedures \cite{Peng2010}, 
which could give a four-body interaction \cite{Luo:20161, Luo:20162} from the combination of two-body interactions and radio-frequency pulses \cite{Peng2010, Tseng1999}.
These systems provide natural platforms for testifying our theoretical studies.

\section*{Acknowledgments}
We would like to thank Ling-Yan Hung and Xueda Wen  for their insightful discussion. C.-T. M would like to thank Nan-Peng Ma for his encouragement. P.-Y. C. was supported by the Rutgers Center for Materials Theory. S.-K. C. was supported by the AFOSR, NSF QIS, ARL CDQI, ARO MURI, ARO and NSF PFC at JQI.


\begin{thebibliography}{99}
\bibitem{EN}
The Euler number of the torus is zero, $\chi= V-E+F=0$ with $V$, $E$, and $F$ being the number of vertices, edges, and faces, respectively.
\end{thebibliography}
\end{document}